\documentclass[12pt,letterpaper,onecolumn,draftcls,journal]{IEEEtran}
\usepackage{epsfig}
\usepackage{cite}
\usepackage{amsfonts}
\usepackage{amsmath}
\usepackage{amssymb}
\usepackage{amsxtra}
\usepackage{float}
\usepackage{rotate}
\usepackage{color}
\usepackage{enumerate}
\usepackage[mathscr]{eucal}
\usepackage[active]{srcltx} 
 \usepackage{Milancommands}

\newcommand{\matriz}[1]{\left[\begin{matrix}#1\end{matrix}\right]}

\newcommand{\norma}[1]{\left| \left| #1 \right| \right|}

\renewcommand{\diag}[1]{\operatorname{diag} \left\{ #1 \right\}  }

\newcommand{\traza}[1]{\operatorname{trace}\left\{ #1 \right\} }

\newcommand{\meanaux}[1]{\mathbb{E}\left\{ #1 \right\} }
\newcommand{\meancr}[2]{\mathcal{E}_{#1}\left\{ #2 \right\} }
\renewcommand{\abs}[1]{\left| #1 \right| }

\newcommand{\treq}[0]{\triangleq}

\newcommand{\evalnoll}[3]{ \lpt #1 \right|_{#2=#3} }

\renewcommand{\D}[0]{ \mathbb{D}}
\newcommand{\R}[0]{ \mathbb{R}}

\renewcommand{\Q}[0]{ \mathbb{Q}}
\newcommand{\N}[0]{ \mathbb{N}}

\renewcommand{\Z}[0]{ \mathbb{Z}}
\newcommand{\C}[0]{ \mathbb{C}}

\newcommand{\igual}[1]{ \stackrel{_{(#1)}}{=} }

\newcommand{\mayori}[1]{ \stackrel{_{(#1)}}{\geq} }

\newcommand{\dec}[0]{ \mathscr{D} }
\newcommand{\enc}[0]{ \mathscr{E} }

\renewcommand{\cal}[1]{\mathcal{#1}}

\def\qed{\relax\ifmmode\hskip2em \Box\Box\Box\else\unskip\nobreak\hskip1em $\Box\Box\Box$\fi}

\newcommand{\fin}{~$\blacksquare$}

\newcommand{\lpt}{\left.}

\newcommand{\lp}{\left(}
\newcommand{\rp}{\right)}

\def\DHLhksqrt#1#2{\setbox0=\hbox{$#1\sqrt{#2\,}$}\dimen0=\ht0
\advance\dimen0-0.2\ht0
\setbox2=\hbox{\vrule height\ht0 depth -\dimen0}%
{\box0\lower0.4pt\box2}}

\newcommand{\BibPath}{/home/milan/Documents/University/Research/BibTeX}

\newtheorem{assumption}{Assumption}[section]
\newtheorem{corollary}{Corollary}[section]

\newtheorem{theorem}{Theorem}[section]
\newtheorem{lemma}{Lemma}[section]
\newtheorem{definition}{Definition}[section]
\newtheorem{remark}{Remark}[section]

\allowdisplaybreaks

\begin{document}
\title{A Characterization of the Minimal Average Data Rate that Guarantees a Given Closed-Loop Performance Level}%
\author{Eduardo~I.~Silva, Milan S. Derpich, Jan \O stergaard and Marco A. Encina
\thanks{
  E.I Silva, M.S. Derpich and M.A. Encina
    are with the Department of Electronic Engineering, Universidad T\'ecnica
         Federico Santa Mar\'ia, Casilla 110-V, Valpara\'iso, Chile
 (emails: 
milan.derpich@usm.cl, marco.encina@alumnos.usm.cl).
Their work was supported by CONICYT through grants FONDECYT
Nr.~1120468, Nr.~1130459, and Anillo ACT-53.}
 \thanks{%
J. {\O}stergaard is with the Department of Electronic Systems,
Aalborg University, Niels Jernes Vej 12, DK-9220, Aalborg,
Denmark (email: janoe@ieee.org). } }

\maketitle

\begin{abstract}
This paper studies networked control systems closed over
noiseless digital channels.  By focusing on noisy LTI plants
with scalar-valued control inputs and sensor outputs, we derive
an absolute lower bound on the minimal average data rate that
allows one to achieve a prescribed level of stationary
performance under Gaussianity assumptions.  We also present a
simple coding scheme that allows one to achieve average data
rates that are at most $1.254$ bits away from the derived lower
bound, while satisfying the performance constraint.  Our
results are given in terms of the solution to a stationary
signal-to-noise ratio minimization problem and builds upon a
recently proposed framework to deal with average data rate
constraints in feedback systems.  A numerical example is
presented to illustrate our findings.
\end{abstract}


\IEEEpeerreviewmaketitle

\section{Introduction}\label{sec:intro}
This paper studies networked control problems for linear
time-invariant (LTI) plants where communication takes place
over a digital communication channel. Such problems have
received much attention in the recent literature
\cite{nafaza07, antbai07}. This interest is motivated by the
theoretical challenges inherent to control problems subject to
data-rate constraints, and by the many practical implications
that the understanding of fundamental limitations in such a
setup may have.

The literature on networked control systems subject to data-rate constraints can be broadly classified into two groups.  A
first group, which includes \cite{delcha90, elimit01, fuxie05,
ishfra02-l, wonbro99, neslib09, brolib00}, uses approaches that
are rooted in nonlinear control theory. An alternative approach
that uses information-theoretic arguments has been adopted in,
e.g., \cite{sahmit06, naieva04, tatmit04a, tasami04, savkin06,
matsav09, chafar08-auto}.  A key question addressed by the
works in the latter group is how to extend, or adapt if
necessary, standard information-theoretic notions to reveal
fundamental limitations in data-rate-limited feedback loops.
Related results have been published in \cite{mardah08,okhais09,
shioht12}, where the interplay between information constraints
and disturbance attenuation is explored.

The most basic question in a data-rate limited feedback control
framework is whether closed-loop stabilization is possible or
not.  
Indeed, stabilization is possible only if the channel data rate
is sufficiently large \cite{wonbro99, bailli02}.  These early
observations spawned several works that study minimal data rate
requirements for stabilization and observability (see, e.g.,
\cite{tatmit04a, tatmit04b, naieva03, fagzam03}). A fundamental
result was presented in \cite{naieva04}.  For noisy LTI plants
controlled over a noiseless digital channels,  it is shown in
\cite{naieva04} that it is possible to find causal coders,
decoders and controllers such that the resulting closed-loop
system is mean-square stable, if and only if the average data
rate is greater than the sum of the logarithm of the absolute
value of the unstable plant poles. A thorough discussion of
this and related work can be found in the survey paper
\cite{nafaza07}. Recent extensions, including stabilization
over time-varying channels, are presented in \cite{youxie11,
micofr13, mifrde09, madael06, yukbas11}.

It is fair to state that stabilization problems subject to data
rate constraints are well-understood.  However, the question of
what is the best closed-loop performance that is achievable
with a given data rate is largely open.  Such problems are
related to causal (and zero-delay) rate-distortion problems
(see, e.g., \cite{neugil82, linzam06, derost12, bomita01,
yuklin12}).  In the latter context, the best results are, to
our knowledge, algorithmic in nature, derived for open-loop
systems and, at times, rely on arbitrarily long delays
\cite{neugil82, linzam06}.  It thus follows that the results in
the above references are not immediately applicable to feedback
control systems.

In the rate-constrained control literature, lower bounds on the
mean-square norm of the plant state have been derived which
show that, when disturbances are present, closed-loop
performance becomes arbitrarily poor when the feedback data
rate approaches the minimum for stability~\cite{naieva04,
nafaza07}.  This result holds no matter how the coder, decoder
and controller are designed. Unfortunately, the bounds in
\cite{naieva04, nafaza07} do not seem to be tight in general.
In contrast, for fully observable noiseless LTI plants with
bounded initial state, \cite{savkin06} shows that one can
(essentially) recover the best non-networked LQR performance
with data rates arbitrarily close to the minimum average data
rate for stabilization. Other results valid in the noiseless or
bounded-support noise cases can be found in, e.g.,
\cite{yukbas06, hunaev05, lemsun06}.

Relevant work on optimal control subject to rate-constraints,
and dealing with unbounded support noise sources, include
\cite{tasami04, nafaza07}.  Those works establish conditions
for separation and certainty equivalence in the context of
quadratic stochastic problems for fully observed plants, when
data rate constraints are present in the feedback path. It is
shown in \cite{nafaza07} that, provided the encoder has a
recursive structure, certainty equivalence and a partial
separation principle hold.  The latter result is relevant.
However, \cite{nafaza07} does not give a practical
characterization of optimal encoding policies.  The results
reported in~\cite{tasami04} share a similar drawback.   Indeed,
performance-related results in \cite{tasami04} are described in
terms of the sequential rate-distortion function, which is
difficult to compute in general.  Moreover, even for the cases
where an expression for such function is available, it is not
clear whether the sequential rate-distortion function is
operationally tight \cite[Section IV-C]{tasami04}. Partial
separation in optimal quantized control problems has been
recently revisited in \cite{fu12}.

Additional results related to the performance of control
systems subject to data-rate constraints are reported
in~\cite{yukbas06, baskjo11} and \cite{mahten09}.  In
\cite{yukbas06}, noiseless state estimation problems subject to
data rate constraints are studied. The case most relevant to
this work uses an asymptotic (in time) quadratic criterion to
measure the state reconstruction error. For such a measure, it
is shown in~\cite{yukbas06} that the bound established in
\cite{naieva04} is sufficient to achieve any prescribed
asymptotic distortion level.  This is achieved, however, at the
expense of arbitrarily large estimation errors for any given
finite time.  On the other hand,~\cite{mahten09} considers
non-linear stochastic control problems over noisy channels, and
a functional (i.e., not explicit) characterization of the
optimal control policies is presented. In turn, \cite{baskjo11}
presents a computationally-intensive iterative method for
encoder and controller design for LTI plants controlled over
noisy discrete memoryless channels. Conditions for separation
and certainty equivalence are also discussed in \cite{baskjo11}
for some specific setups.

In this paper, we focus on the feedback control of noisy LTI
plants, with one-dimensional control inputs and sensor outputs,
that are controlled over a noiseless (and delay-free) digital
channel. By considering causal but otherwise unconstrained
coding schemes, we study the minimal (operational) average data
rate, say $\mathscr{R}(D)$, that guarantees that the
steady-state variance of an error signal is below a
prespecified level $D>0$.   By assuming that the plant initial
state and the disturbances are jointly Gaussian, our first
contribution is to show that a lower bound on $\mathscr{R}(D)$
can be obtained by minimizing the directed information rate
\cite{massey90} across an auxiliary zero-delay coding scheme
that behaves as an LTI system plus additive white Gaussian
noise. For doing so, we build upon \cite{sideos11-TAC} and make
use of information-theoretic arguments that complement previous
results in \cite{mardah05, mardah08}. Motivated by our first
result, and as a second contribution, we generalize the class
of randomized coding schemes proposed in \cite{sideos11-TAC}
(see also \cite{zamfed92}) and use the coding schemes so
obtained, to characterize an upper bound on $\mathscr{R}(D)$.
Whilst not tight in general, the gap between the derived upper
and lower bounds is smaller than (approximately) $1.254$ bits
per sample. Our results are constructive and given in terms of
the solution to a signal-to-noise ratio (SNR) constrained
optimal control problem (see also \cite{brmifr07,
sigoqu10-auto,johann11}). We also propose a specific randomized
coding scheme that achieves the prescribed level of performance
$D$, while incurring an average data rate that is strictly
smaller than the derived upper bound on $\mathscr{R}(D)$.

This paper extends our works \cite{sideos11-TAC, sideos11-TN}
in at least two respects.  First, this paper considers LTI
plants that are not constrained to be stabilizable by unity
feedback (or said otherwise, we do not exploit any predesigned
controller for the plant). 
Second, we construct a universal
lower bound on the minimal average data rate that guarantees a
prescribed performance level which cannot be derived from the
arguments used in \cite{sideos11-TAC}. Indeed, the results in
\cite{sideos11-TAC} and \cite{sideos11-TN} are valid only when
a specific class of source coding schemes is employed. Here, we
do not, \emph{a priori}, constrain the type, structure or
complexity of the considered source coding schemes.

The remainder of this paper is organized as follows: Section
\ref{sec:setup-prob-def} describes the problem addressed in the
paper.  
Section~\ref{sec:lower} presents a lower bound on the
minimal average data rate that guarantees a given performance
level, whilst Section \ref{sec:upper-linear-coding} presents
the corresponding upper bound.  Section~\ref{sec:computing-and-approximations} discusses how to solve
the related SNR-constrained control problem characterizing both
our upper and lower bounds, and comments on implementation
issues. Finally, Section~\ref{sec:ejemplo} presents a numerical
example, and Section \ref{sec:conclusions} draws conclusions.
Early versions of part of the results in this paper were
reported in \cite{sideos10-necsys}.

\emph{Notation:}  $\R$ denotes the set of real numbers, $\R^+$
denotes the set of strictly positive real numbers, $\R_0^+\treq
\R^+\cup\{0\}$, $\N_0\treq \{0,1,\cdots\}$.  In this paper,
$\log$ stands for natural logarithm, and $\abs{x}$ for the
magnitude (absolute value) of $x$.  We work in discrete time
and use $k$ for the time index.  An LTI filter $X$ is said to
be proper (i.e., causal) if its transfer function $X(z)$
remains finite when $z\to\infty$, and it is said biproper if it
is proper and $\lim_{z\to\infty} X(z) \neq 0$.  We define the
set $\cal{U}_{\infty}$ as the set of all proper and stable
filters with inverses that are also stable and proper.


In this paper, all random processes are defined for $k\in\N_0$.
All random variables and processes are assumed to be
vector-valued, unless stated otherwise. Given a process $x$, we
denote its $k^{th}$ sample by $x(k)$ and use $x^k$ as shorthand
for $x(0), \dots, x(k)$. We say that a random process is a
second-order one if it has first- and second-order moments that
are bounded for every $k$ and that also remain bounded as
$k\to\infty$. Gaussian processes are, by definition,
second-order ones \cite{doob53}. We use $\mathbb{E}$ to denote
the expectation operator.  A process $x$ is said to be
asymptotically wide-sense stationary (AWSS) if and only if
there exist $\mu_x$ and a function $R_x(\tau)$, both
independent of the statistics of $x(0)$, such that
$\lim_{k\to\infty} \meanaux{x(k)}=\mu_x$ and $\lim_{k\to\infty}
\mathbb{E} \{ \lp x(k+\tau) - \meanaux{x(k+\tau)} \rp \lp x(k)
- \meanaux{x(k)} \rp^{T} \} = R_x(\tau)$ for every
$\tau\in\N_0$. The steady-state spectral density of an AWSS
process is denoted by $S_x$ (and defined as the Fourier
transform of $R_x(\tau)$ extended for $\tau<0$ according to
$R_x(\tau)=R_x(-\tau)^T)$.  The corresponding steady-state
covariance matrix is denoted by $P_x$, and $\sigma_x^2 \treq
\traza{P_x}$.  Jointly second-order and jointly AWSS processes
are defined in the obvious way.  Appendix \ref{ap:info-theory}
recalls some useful notation and results from Information
Theory \cite{covtho06}.

\section{Problem Setup}\label{sec:setup-prob-def}
This paper focuses on the networked control system (NCS) of
Figure~\ref{fig:ncs-general}. In that figure, $P$ is an LTI
plant, $u\in\R$ is the control input, $y\in\R$ is a sensor
output, $e\in\R^{n_e}$ is a signal related to closed-loop
performance, and $d\in\R^{n_d}$ is a disturbance. The feedback
path in Figure~\ref{fig:ncs-general} comprises a digital
channel and thus quantization becomes mandatory. This task is
carried out by an encoder whose output corresponds to a
sequence of binary words. 
These words are then transmitted over
the channel, and mapped back into real numbers by a decoder.
The encoder and decoder also embody a controller for the plant.

\begin{figure}
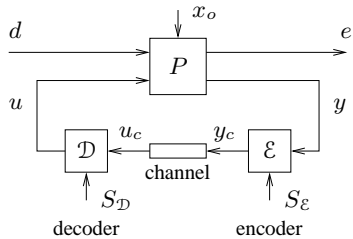

\centering
\input ncs.pstex_t
\caption{Networked control system where communication takes place over a digital channel.}\label{fig:ncs-general}
\end{figure}

We partition $P$ in a way such that
\begin{align}\label{eq:P-partida}
\matriz{ e \\ y } = \matriz{ P_{11} & P_{12} \\ P_{21} & P_{22} } \matriz{ d \\ u },
\end{align}
where $P_{ij}$ are proper transfer functions of suitable
dimensions.  We will make use of the following assumptions.

\begin{assumption}\label{assu:G}
$P$ is a proper LTI plant, free of unstable hidden modes, such
that the open-loop transfer function from $u$ to $y$ (i.e.,
$P_{22}$ in \eqref{eq:P-partida}) is single-input single-output
and strictly proper.  The initial state of the plant, say
$x_o$, and the disturbance $d$ are jointly Gaussian, $d$ is
zero-mean white noise with unit variance $P_d=I$, and $x_o$ has
finite differential entropy (i.e., the variance of $x_o$ is
positive definite). \fin
\end{assumption}

We focus on error-free zero-delay digital channels and denote
the channel input alphabet by $\cal{C}$, a countable set of
prefix-free binary words \cite{covtho06}.   Whenever the
channel input symbol $y_c(k)$ belongs to $\cal{C}$, the
corresponding channel output is given by $u_c(k)=y_c(k)$.   The
expected length of $y_c(k)$ is denoted by $R(k)$, and the
average data rate across the channel is thus defined
as\footnote{We measure $\mathscr{R}$ in nats per sample (recall
that $\log 2$ nats correspond to $1$ bit).}
\begin{align}
\mathscr{R}\treq \lim_{k\to\infty} \frac{1}{k} \sum_{i=0}^{k-1} R(i).
\end{align}

We assume the encoder to be an arbitrary (hence possibly
nonlinear and time-varying) causal system such that the channel
input $y_c$ satisfies
\begin{align}\label{eq:enc}
y_c(k) = \enc_k(y^k, S_{\enc}^k),
\end{align}
where $\alpha^k$ is shorthand for $\alpha(0), \dots,
\alpha(k)$, $S_{\enc}(k)$ denotes side information that becomes
available at the encoder at time instant $k$, and $\enc_k$ is a
(possibly nonlinear and time-varying) deterministic mapping
whose range is a subset of $\cal{C}$.  Similarly, we assume
that the decoder is such that the plant input $u$ is given by
\begin{align}\label{eq:dec}
u(k) =  \dec_k(u_c^k, S_{\dec}^k),
\end{align}
where $S_{\dec}(k)$ denotes side information that becomes
available at the decoder at time instant $k$, and $\dec_k$ is a
(possibly non-linear and time-varying) deterministic mapping.

\begin{assumption}\label{assu:coding-schemes}
The systems $\enc$ and $\dec$ in Figure \ref{fig:ncs-general}
are causal, possibly time-varying or non-linear, described by \eqref{eq:enc}--\eqref{eq:dec}.
The side information sequences $S_{\enc}$ and $S_{\dec}$ are
jointly independent of $(x_o, d)$, and the decoder is
invertible upon knowledge of $u^i$ and $S_\dec^i$, i.e.,
$\forall i\in\N_0$, there exists a deterministic mapping $g_i$
such that $u_c^{i}=g_i(u^{i}, S_{\dec}^{i})$. \fin
\end{assumption}

The assumption on the side information sequences is motivated
by the requirement that (causal) encoders and decoders use only
past and present input values, and additional information not
related to the message being sent, to construct their current
outputs (see also page 5 in \cite{massey90}). 
On the other
hand, if, for some encoder $\enc$ and decoder $\dec$, the
decoder is not invertible, then one can always define an
alternative encoder and decoder pair, where the decoder is
invertible, yielding the same input-output relationship as
$\enc$ and $\dec$, but incurring a lower average data rate
\cite[Lemma 4.1]{sideos11-TAC}. Accordingly, one can focus,
without loss of generality, on encoder-decoder pairs where the
decoder is invertible.

In this paper, we adopt the following notion of stability (see
also \cite{cofrma05}):

\begin{definition}\label{def:estabilidad}
We say that the NCS of Figure \ref{fig:ncs-general} is
asymptotically wide-sense stationary (AWSS) if and only if the
state of the plant $x$, the output $y$, the control input $u$,
and the disturbance $d$, are jointly second-order AWSS
processes. \fin
\end{definition}

\begin{remark}
The notion of stability introduced above is stronger than the
usual notion of mean-square stability (MSS) where only
$\sup_{k\in \N_0} \meanaux{ x(k)x(k)^T } < \infty$ is required
to hold (see, e.g., \cite{naieva04}). 
\fin
\end{remark}

The goal of this paper is to characterize, for the NCS of
Figure \ref{fig:ncs-general}, the minimal average data rate
$\mathscr{R}$ that guarantees a given performance level as
measured by the steady-state variance of the output $e$.  We
denote by $D_{\inf}$ the infimal steady-state variance of $e$
that can be achieved by setting $u(k)=\mathscr{K}_k(y^k)$, with
$\mathscr{K}_k$ being a (possibly nonlinear and time-varying)
deterministic mapping, under the constraint that the resulting
feedback loop AWSS.  With this definition, we formally state
the problem of interest in this paper as follows:  Find, for
any $D\in (D_{\inf}, \infty)$ and whenever Assumption
\ref{assu:G} holds,\footnote{In this paper we adhere to the
convention that an unfeasible minimization problem has an
optimal value equal to $+\infty$ \cite{boyvan04}.}
\begin{align}\label{eq:def-RD}
\mathscr{R}(D) \treq \inf_{ \sigma_e^2\leq D } \mathscr{R},
\end{align}
where $\sigma_e^2\treq \traza{P_e}$, $P_e$ is the steady-state
covariance matrix of $e$, and the optimization is carried out
with respect to all encoders $\enc$ and decoders $\dec$ that
satisfy Assumption \ref{assu:coding-schemes} and render the
resulting NCS AWSS.

It can be shown that the problem in \eqref{eq:def-RD} is
feasible for every $D\in (D_{\inf}, \infty)$ (see Appendix
\ref{ap:factibilidad}). If $D<D_{\inf}$, then the problem is
clearly unfeasible. On the other hand, achieving $D=D_{\inf}$
incurs an infinite average data rate, except for very special
cases. We will thus focus on $D\in (D_{\inf}, \infty)$ without
loss of generality.

The remainder of this paper characterizes $\mathscr{R}(D)$
within a gap smaller than (approximately) $1.254$ bits per
sample. Such characterization is given in terms of the solution
to a constrained quadratic optimal control problem.  We also
propose encoders and decoders which achieve an average data
rate within the above gap, while satisfying the performance
constraint on the steady-state variance of $e$.

\section{An Information-Theoretic lower bound on $\mathscr{R}(D)$}\label{sec:lower}
This section shows that a lower bound on $\mathscr{R}(D)$ can
be obtained by minimizing the directed information rate across
an auxiliary coding scheme comprised of LTI systems and an
additive white Gaussian noise channel with feedback.  The
starting point of our presentation is a result in
\cite{sideos11-TAC}.

\begin{theorem}[Theorem 4.1 in \cite{sideos11-TAC}]\label{teo:rates-directed}
Consider the NCS of Figure \ref{fig:ncs-general} and suppose
that Assumptions \ref{assu:G} and \ref{assu:coding-schemes}
hold. Then,
\begin{align}\label{eq:def-directed-info}
\mathscr{R} \geq I_{\infty}( y \to u) \treq
\lim_{k\to\infty} \frac{1}{k} \sum_{i=0}^{k-1} I( u(i); y^i | u^{i-1} ),
\end{align}
where $I( \, \cdot \, ; \, \cdot \, | \, \cdot \, )$ denotes
conditional mutual information (see Appendix
\ref{ap:info-theory}). \fin
\end{theorem}

The quantity $I_{\infty}( y \to u)$ corresponds to the directed
information rate \cite{massey90} across the source coding
scheme of Figure \ref{fig:ncs-general} (i.e., between the input
$y$ and the output $u$ of the source coding scheme).  Note that
$I_{\infty}( y \to u)$ is a function of the joint statistics of
$y$ and $u$ only.

We will now derive a lower bound on the directed information
rate across the considered coding scheme, in terms of the
directed information rate that would appear if all the involved
signals were Gaussian.

\begin{lemma}\label{lema:cota-gauss}
Consider the NCS of Figure \ref{fig:ncs-general} and suppose
that Assumptions \ref{assu:G} and \ref{assu:coding-schemes}
hold.  If, in addition, $(x_o,d,y,u)$ are jointly second-order,
then $I_{\infty}( y \to u ) \geq I_{\infty}( y_G \to u_G)$,
where $y_G$ and $u_G$ are such that $(x_o,d,y_G,u_G)$ are
jointly Gaussian with the same first- and second-order (cross-)
moments as $(x_o,d,y,u)$.
\end{lemma}

\begin{proof}
Our claim follows from the following chain of equalities and
inequalities:
\begin{align}\label{eq:dirigida}
\sum_{i=0}^{k-1} I( u(i); y^i | u^{i-1} ) \igual{a} I(x_o, d^{k-1}; u^{k-1})
    \mayori{b} I(x_o, d^{k-1}; u_{G}^{k-1})
    \igual{c}  \sum_{i=0}^{k-1} I( u_G(i); y_G^i | u_G^{i-1} ),
\end{align}
where $(a)$ follows from Assumption \ref{assu:coding-schemes}
and Lemma \ref{lema:directed-estado} with $(x_{1,o},
d_1)=(x_o,d)$, $\bar{y}=y$, $\bar{u}=u$ and $(x_{2,o},
d_2)=(S_{\dec}, S_{\enc})$, $(b)$ follows from Lemma
\ref{lema:gauss-minimiza} in Appendix \ref{ap:info-theory}, and
$(c)$ follows by using Lemma \ref{lema:directed-estado} again.
The result is now immediate from \eqref{eq:def-directed-info}
and \eqref{eq:dirigida}.\fin
\end{proof}

It follows from Theorem \ref{teo:rates-directed} and Lemma
\ref{lema:cota-gauss} that, in order to bound $\mathscr{R}(D)$
from below, it suffices to minimize the directed information
rate that would appear across the source coding scheme of
Figure \ref{fig:ncs-general}, when its input $y$ and output $u$
are jointly Gaussian AWSS processes.\footnote{Note that, given
our definition of $\mathscr{R}(D)$, our focus is precisely on
encoders and decoders that render $(x,y,u,d)$, and hence also
$(y,u)$, jointly second-order AWSS.}

\begin{lemma}\label{lema:directed-gaussiana}
Assume that $u$ and $y$ are jointly Gaussian AWSS processes.
Then,
\begin{align}
I_{\infty}(y\to u) = \frac{1}{4\pi} \int_{-\pi}^{\pi} \log{ \lp \frac{S_u \ejw }{\sigma_n^2} \rp } \, d\omega,
\end{align}
where $S_u$ is the steady-state power spectral density of $u$,
and $\sigma_n^2$ is the steady-state variance of the Gaussian
AWSS sequence of independent random variables $n$, defined via
\begin{align}\label{eq:def-n}
n(k) \treq u(k) - \hat{u}(k), \quad \hat{u}(k) \treq \meanaux{u(k) | y^k, u^{k-1} }.
\end{align}
\end{lemma}

\begin{proof}
We start by noting that, since $(u,y)$ are jointly Gaussian
AWSS processes, a simple modification of the proof of Theorem
2.4 in \cite[p.~20]{porat94} yields the conclusion that $n$ is
also Gaussian and AWSS.

To proceed, we note that
\begin{align}\label{eq:mutual-info-Gauss}
I(u(i); y^i | u^{i-1}) & \igual{a} h(u(i)|u^{i-1}) - h(u(i)| y^i, u^{i-1}) \nonumber \\
                       & \igual{b} h(u(i)|u^{i-1}) - h( n(i) + \hat{u}(i) | y^i, u^{i-1}) \nonumber \\
                       & \igual{c} h(u(i)|u^{i-1}) - h( n(i) | y^i, u^{i-1}) \nonumber \\
                       & \igual{d} h(u(i)|u^{i-1}) - h( n(i) ),
\end{align}
where $(a)$ follows from Property 1 in Appendix
\ref{ap:info-theory}, $(b)$ follows from the definition of
$\hat{u}$, $(c)$ follows from Property 2 in Appendix
\ref{ap:info-theory} and the fact that, by construction,
$\hat{u}(i)$ is a deterministic function of $(y^i, u^{i-1})$,
and $(d)$ follows from Property 3 in Appendix
\ref{ap:info-theory} and the fact that (again by construction),
$n(i)$ is independent of $(y^i, u^{i-1})$. Now,
\eqref{eq:mutual-info-Gauss} and the definition of directed
information rate yields
\begin{align}\label{eq:mutual-info-Gauss-listo}
I_{\infty}(y \to u)
    & =  \lim_{k\to \infty} \frac{1}{k}   \sum_{i=0}^{k-1} \Big\{ h(u(i)|u^{i-1})
                    - h(n(i)) \Big\} \nonumber \\
    & \igual{a} \lim_{k\to \infty} \frac{1}{k} \lp h(u^{k-1}) - h(n^{k-1}) \rp \nonumber \\
    & \igual{b} \frac{1}{4\pi} \int_{-\pi}^{\pi} \log{\lp 2\pi e S_u \ejw \rp }\, d\omega - \frac{1}{2} \log{\lp 2\pi e  \sigma_n^2 \rp},
\end{align}
where $(a)$ follows from Properties 3 and 4 in Appendix
\ref{ap:info-theory} and the fact that, by construction, $n(k)$
is independent of $n^{k-1}$, and $(b)$ follows from Lemma 4.3
in \cite{madado07} and the fact that both $u$ and $n$ are
Gaussian and AWSS. The result is now immediate from
\eqref{eq:mutual-info-Gauss-listo}.\fin
\end{proof}

Lemma~\ref{lema:directed-gaussiana} characterizes the directed
information rate between Gaussian AWSS processes in terms of
the spectrum of the process towards which the mutual
information is directed.   
Lemma~\ref{lema:directed-gaussiana}
generalizes Theorem 4.6 in~\cite{elia04}, where the author
calculates directed information rates between Gaussian
processes that are linked by an additive white Gaussian noise
channel.

\begin{figure}
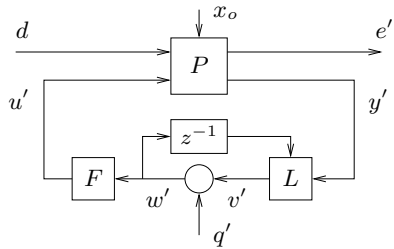

\centering
\input lazo_con_linear_coding.pstex_t
\caption{Auxiliary LTI system that arises when the encoder and decoder of Figure \ref{fig:ncs-general}
are replaced by proper LTI filters $F$ and $L$ and an additive white noise channel with (one-step delayed) feedback.}\label{fig:lazo-con-linear-coder}
\end{figure}

We are now ready to present the main result of this section. To
that end, we begin by noting that Theorem~\ref{teo:rates-directed} and Lemma~\ref{lema:cota-gauss}
readily imply that for any encoder and decoder satisfying
Assumption~\ref{assu:coding-schemes}, and rendering the
resulting NCS AWSS, $\mathscr{R}\geq I_{\infty}(y_G\to u_G)$,
where $(y_G,u_G)$ are such that $(x_o,d,u_G,y_G)$ are jointly
Gaussian with the same first- and second-order (cross-) moments
as $(x_o,d,u,y)$. We also note that that the adopted
performance measure is quadratic.  The above observations imply
that one can always match (or improve) the rate-performance
tradeoff of a given encoder-decoder pair by choosing, instead,
an encoder and a decoder which, besides rendering the NCS AWSS,
renders $(u, y)$ jointly Gaussian with $(x_o, d)$.   Since the
plant initial state and the disturbance are Gaussian, and the
plant is LTI, one possible way of achieving such pair of
signals $(u,y)$ is by using the LTI feedback architecture of
Figure \ref{fig:lazo-con-linear-coder}.  We formalize these
observations below.

Define the auxiliary LTI feedback scheme of Figure
\ref{fig:lazo-con-linear-coder}, where everything is as in
Figure \ref{fig:ncs-general} except for the fact that we have
replaced the link between the plant output $y$ and the plant
input $u$ by a set of proper LTI filters, $F$ and $L$, and an
additive noise channel with (one-step delayed) feedback and
noise $q'$ such that
\begin{align}\label{eq:linear-coding-explicito}
u' = F w', \quad w' = v' + q', \quad v' = L \diag{z^{-1}, 1} \matriz{w'\\ y'},
\end{align}
where $z^{-1}$ stands for the unit delay.  
In Fig.~\ref{fig:lazo-con-linear-coder}, we assume that the plant $P$,
the disturbance $d$ and the plant initial state $x_o$ satisfy
Assumption~\ref{assu:G}, that the initial states of $F,L$ and
of the delay are deterministic, and that $q'$ is zero mean
Gaussian white noise, independent of $(x_o,d)$, and having
constant variance $\sigma_{q'}^2$.

In Fig.~\ref{fig:lazo-con-linear-coder}, we have added
apostrophes (as in $e'$) to all symbols that refer to signals
that have a counterpart in the scheme of Fig.~\ref{fig:ncs-general} with possibly different statistics. To
streamline our presentation, we adopt the convention that,
whenever we refer to the auxiliary feedback system of Figure
\ref{fig:lazo-con-linear-coder}, it is to be understood that we
are implicitly working under the assumptions stated in the
above paragraph.

\begin{theorem}\label{teo:lower-bound}
Consider the NCS of Figure \ref{fig:ncs-general} and suppose
that Assumptions \ref{assu:G} and \ref{assu:coding-schemes}
hold.  If $D\in (D_{\inf}, \infty)$, then
\begin{align}\label{eq:lower-bound}
\mathscr{R}(D) \geq \phi_u'(D), \quad
    \phi_u'(D) \treq  \inf_{\sigma_{e'}^2 \leq D} \, \frac{1}{4\pi}\int_{-\pi}^{\pi}
            \log{ \lp \frac{S_{u'} \ejw }{\sigma_{q'}^2} \rp  } \,  d\omega ,
\end{align}
where the optimization defining $\phi_u'(D)$ is performed with
respect to all proper LTI filters $L$, and auxiliary noise
variances $\sigma_{q'}^2\in\R^+$, that render the LTI feedback
system of Figure \ref{fig:lazo-con-linear-coder} with $F=1$
internally stable and well-posed, and $S_{u'}$ and
$\sigma_{e'}^2$ denote the steady-state power spectral density
of $u'$ and the steady state variance $e'$ in Figure
\ref{fig:lazo-con-linear-coder}, respectively.
\end{theorem}

\begin{proof}
Denote by $\cal{C}_D$ the set of all encoders $\enc$ and
decoders $\dec$ that satisfy Assumption
\ref{assu:coding-schemes}, render the NCS of Figure
\ref{fig:ncs-general} AWSS and guarantee that $\sigma_e^2 \leq
D$.   
Also, denote by $\cal{C}_{D,
G}$ the subset of $\cal{C}_D$ containing all encoders $\enc$
and decoders $\dec$ that in addition render $u$ and $y$ jointly
Gaussian. 
(Since $D>D_{\inf}$, $\cal{C}_{D,G}$ is non empty, and hence $\cal{C}_{D}$ is non empty as well;  see
Appendix \ref{ap:factibilidad}.)
Given Theorem \ref{teo:rates-directed}, the
definition of $\mathscr{R}(D)$, and the fact that $D>D_{\inf}$
guarantees that the problem of finding $\mathscr{R}(D)$ is
feasible, it follows that
\begin{align}\label{eq:primera-cota-optimo-RD}
\mathscr{R}(D) & \geq \inf_{(\enc,\dec) \in \; \cal{C}_D} I_{\infty}(y \to u) \nonumber \\
               & \mayori{a} \inf_{(\enc,\dec) \in \; \cal{C}_{D, G}} I_{\infty}(y \to u) \nonumber \\
               & \igual{b}  \inf_{(\enc,\dec) \in \; \cal{C}_{D, G}} \;
                        \frac{1}{4\pi} \int_{-\pi}^{\pi} \log{ \lp \frac{S_u \ejw }{\sigma_n^2} \rp } \, d\omega,
\end{align}
where $(a)$ follows from Lemma \ref{lema:cota-gauss}, and $(b)$
follows from Lemma \ref{lema:directed-gaussiana}.

To proceed, pick any $(\enc,\dec) \in \; \cal{C}_{D, G}$ and
recall the definition of the noise source $n$ in
\eqref{eq:def-n}.   By definition of $\cal{C}_{D,G}$,
\eqref{eq:def-n} is equivalent to the existence of a sequence
of linear mappings $L_k$, $k\in\N_0$, such that
\begin{align}\label{eq:mapa-L}
u(k) = L_k(y^k, u^{k-1}) + n(k),
\end{align}
where $n(k)$ is independent of $(y^k, u^{k-1})$.  Since, for
$(\enc,\dec) \in \; \cal{C}_{D, G}$, $(y,u)$ are jointly AWSS,
it follows from a straightforward modification of the material
in \cite[p.~19]{porat94} that $L_k$ converges to an LTI mapping
as $k\to\infty$.  Such limiting mapping renders the resulting
NCS internally stable and well-posed (otherwise the underlying
encoder and decoder would not be in $\cal{C}_{D, G}$), and
defines the steady-state spectrum $S_u$ of $u$ and the
steady-state variances $\sigma_n^2$ and $\sigma_e^2$ of both
$n$ and $e$. (Here, we use the fact that $n$ is also AWSS; see
proof of Lemma \ref{lema:directed-gaussiana}.)

Now, consider the auxiliary LTI feedback system of Figure
\ref{fig:lazo-con-linear-coder} described before.   Assume that
$F=1$, that $L$ reproduces the steady-state behavior of $L_k$
in \eqref{eq:mapa-L}, and that $q'$ has a variance equal to
$\sigma_n^2$ (see previous paragraph).  With the above choices
for $F$, $L$ and $q'$, and given the properties of the limiting
map $L_k$ summarized in the above paragraph, it follows that
the feedback system of Figure \ref{fig:lazo-con-linear-coder}
is internally stable and well-posed and, in particular, that
the plant input $u'$ admits a steady-state power spectral
density $S_{u'}$ that, by construction, equals $S_u$ in the
previous paragraph. Similarly, the error signal $e'$ in Figure
\ref{fig:lazo-con-linear-coder} admits a steady-state variance
$\sigma_{e'}^2$ that equals $\sigma_e^2$.  By mirroring the
derivations leading to \eqref{eq:mutual-info-Gauss-listo} it
thus follows that
\begin{align}\label{eq:mutual-tilde-igual-mutual}
I_{\infty}(y'\to u') & =  \frac{1}{4\pi} \int_{-\pi}^{\pi} \log{ \lp \frac{S_{u'} \ejw }{\sigma_{q'}^2} \rp } \, d\omega
        = \frac{1}{4\pi} \int_{-\pi}^{\pi} \log{ \lp \frac{S_{u} \ejw
}{\sigma_{n}^2} \rp } \, d\omega.
\end{align}
We thus conclude that, for any encoder and decoder in
$\cal{C}_{D, G}$, there exist a proper LTI filter $L$ and a
Gaussian white noise source $q'$ such that, when $F=1$, the
mutual information rate $I_{\infty}( y'\to u')$ in Figure
\ref{fig:lazo-con-linear-coder} equals $I_{\infty}(y\to u)$ in
Figure \ref{fig:ncs-general} while achieving
$\sigma_{e'}^2=\sigma_e^2$.  Our claim is now immediate from
\eqref{eq:primera-cota-optimo-RD},
\eqref{eq:mutual-tilde-igual-mutual} and the properties of
$L$.\fin
\end{proof}

Theorem \ref{teo:lower-bound} states that a lower bound on the
minimal average data rate that guarantees a given performance
level, can be obtained by solving an optimization problem which
is stated for the auxiliary LTI feedback system of Figure
\ref{fig:lazo-con-linear-coder}, where communication takes
place over an additive white Gaussian noise channel with
feedback. 

We finish this section by deriving a simpler lower bound on
$\mathscr{R}(D)$.   To that end, we will first state an
auxiliary result.

\begin{lemma}\label{lema:w-blanco}
Consider the LTI feedback system of Figure
\ref{fig:lazo-con-linear-coder}.  Fix $\sigma_{q'}^2\in\R^+$
and define (whenever the involved quantities exist)
\begin{align}\label{eq:def-phi-w}
\phi_w'(F,L, \sigma_{q'}^2) \treq \frac{1}{4\pi}\int_{-\pi}^{\pi} \log \lp
\frac{S_{w'} \ejw }{\sigma_{q'}^2} \rp \, d\omega,
\end{align}
where $S_{w'}$ is the steady-state power spectral density of
$w'$.  If the pair $(F,L) = (F^{(0)}, L^{(0)})$ renders the
feedback system of Figure \ref{fig:lazo-con-linear-coder}
internally stable and well-posed, then there exist a second
pair of filters, namely $(F, L) = (F^{(1)}, L^{(1)})$, with
$F^{(1)}$ biproper, that also defines an internally stable and
well-posed feedback loop, leaves the steady-state power
spectral density of $e'$ unaltered, and is such
that\footnote{The notation $\left. X \right|_{Z=Z_1}$ is used
to denote the quantity $X$ when $Z=Z_1$.}
\begin{align}\label{eq:clave-w-blanco}
\phi_w'(F^{(1)}, L^{(1)}, \sigma_{q'}^2) = \phi_w'(F^{(0)}, L^{(0)}, \sigma_{q'}^2) = \left. \frac{1}{2}\log{\lp 1 +
\frac{\sigma_{v'}^2}{\sigma_{q'}^2} \rp} \right|_{(F, L)=(F^{(1)}, L^{(1)})} - \eta
\end{align}
for any (arbitrarily small) $\eta>0$.
\end{lemma}

\begin{proof}
Consider Figure \ref{fig:lazo-con-linear-coder} and the
partition for $P$ in \eqref{eq:P-partida}.  Introduce proper
transfer functions $L_y$ and $L_w$ such that $L = [ \, L_w \;
L_y \,]$ (see \eqref{eq:linear-coding-explicito}). A standard
argument \cite{franci87} shows that the feedback system of
Figure \ref{fig:lazo-con-linear-coder} is internally stable and
well-posed if and only if the transfer function $T$ between
$[\, q' \; d \; n_1 \; n_2 ]^T$ and $[\, e' \; y' \; w' \; u'
]^T$ in Figure \ref{fig:lazo-con-linear-coder-para-estabilidad}
is stable and proper. It is straightforward to see that
\begin{align}
T = \matriz{ P_{12}FS &  P_{11} + P_{12}FL_ySP_{21} & P_{12} \lp 1 - \frac{L_w}{z} \rp S  & P_{12}FL_yS \\
P_{22}FS & \lp 1 - \frac{L_w}{z} \rp P_{21} S & \lp 1 - \frac{L_w}{z} \rp P_{22}S  & P_{22}FL_yS \\
S & L_yP_{21}S & L_yP_{22}S & L_yS \\
FS & FL_ySP_{21} & \lp 1 - \frac{L_w}{z} \rp S & FL_y S },
\end{align}
where
\begin{align}\label{eq:def-S}
S\treq \lp 1 - L_wz^{-1} - P_{22}F L_y \rp^{-1}.
\end{align}
We will write $T^{(i)}$ to refer to the matrix $T$ that arises
when $(F,L)=(F^{(i)}, L^{(i)})$, $i\in\{0,1\}$.  Similarly,
$L_y^{(i)}$ and $L_w^{(i)}$ refer to the components of $L$,
when $L=L^{(i)}$.

\begin{figure}
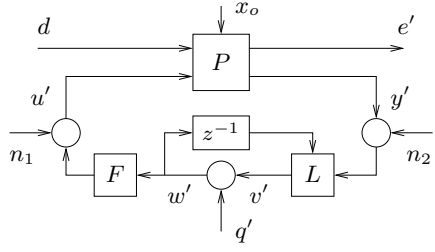

\centering
\input lazo_con_linear_coding_estabilidad.pstex_t
\caption{Auxiliary feedback system for stability analysis.}
\label{fig:lazo-con-linear-coder-para-estabilidad}
\end{figure}

Set
\begin{align}\label{eq:alternativos}
F^{(1)} = z^{n_0} F^{(0)} X^{-1}, \quad L_y^{(1)}= z^{-n_0} L_y^{(0)},
    \quad L_w^{(1)} = z \lp 1 - \lp 1-\frac{L_w^{(0)}}{z}\rp X^{-1} \rp,
\end{align}
where $n_0$ is the relative degree of $F^{(0)}$,
$X\in\cal{U}_{\infty}$ and $X(\infty)=1$. Given
\eqref{eq:alternativos} and the fact that $X\in
\cal{U}_{\infty}$, it follows that $F^{(1)}$ is biproper and
that
\begin{align}
T^{(1)}=\diag{z^{n_0}I,z^{n_0}I,X,z^{n_0}I}T^{(0)}\diag{I, z^{-n_0}I, z^{-n_0}I, z^{-n_0}I}.
\end{align}
The definition of $n_0$ and $X$ guarantees that the pair
$(F^{(1)}, L^{(1)})$ renders the feedback system of Figure
\ref{fig:lazo-con-linear-coder} internally stable and
well-posed if and only if $(F^{(0)}, L^{(0)})$ does so.  It
also immediately follows that $(F^{(1)}, L^{(1)})$ defines the
same stationary spectral density for $e'$ than $(F^{(0)},
L^{(0)})$.

To complete the proof, we now propose specific choice for $X$.
Denote by $w'^{(i)}$ the signal $w'$ that arises when $(F,L) =
(F^{(i)}, L^{(i)})$. Write
$S_{w'^{(0)}}=\abs{\Omega_{w'^{(0)}}}^2$, where
$\Omega_{w'^{(0)}}$ is stable, biproper, and has all its zeros
in $\{z\in\C: \abs{z}\leq 1\}$. Denote by $c_1,\dots,c_{n_c}$
the zeros of $\Omega_{w'^{(0)}}$ that lie on the unit circle.
Define, for $\epsilon \in (0,1)$,
\begin{align}
\tilde{\Omega}_{w'^{(0)}} \triangleq \Omega_{w'^{(0)}}
\prod_{i=1}^{n_c} z(z-c_i)^{-1}, \quad X_{\epsilon}\triangleq \lp
\tilde{\Omega}_{w'^{(0)}} \rp^{-1}\tilde{\Omega}_{w'^{(0)}}(\infty)
\prod_{i=1}^{n_c} z(z-\epsilon
c_i)^{-1}.
\end{align}
By construction, $X_{\epsilon}(\infty)=1$ and
$X_{\epsilon}\in\cal{U}_{\infty}$ for every $\epsilon\in(0,1)$.
It now follows, by proceeding as in the proof of Theorem 5.2 in
\cite{sideos11-TAC}, that there exists $\epsilon\in(0,1)$ such
that setting $X=X_{\epsilon}$ in~\eqref{eq:alternativos}
guarantees that $(F^{(1)}, L^{(1)})$ is such that
\eqref{eq:clave-w-blanco} holds for any $\eta>0$. \fin
\end{proof}

\begin{corollary} \label{coro:lower-bound-snr}
Consider the NCS of Figure \ref{fig:ncs-general} and suppose
that Assumptions \ref{assu:G} and \ref{assu:coding-schemes}
hold.   If $D\in (D_{\inf}, \infty)$, then
\begin{align}\label{eq:lower-bound-snr}
\mathscr{R}(D) \geq \frac{1}{2} \log{ \lp 1 + \gamma'(D) \rp }, \quad
\gamma'(D) \treq \inf_{\sigma_{e'}^2 \leq D} \; \gamma', \quad \gamma' \treq \frac{\sigma_{v'}^2}{\sigma_{q'}^2},
\end{align}
where the optimization defining $\gamma'(D)$ is performed with
respect to all proper LTI filters $F$ and $L$, and auxiliary
noise variances $\sigma_{q'}^2\in\R^+$, that render the LTI
feedback system of Figure \ref{fig:lazo-con-linear-coder}
internally stable and well-posed, and $\sigma_{v'}^2$ and
$\sigma_{e'}^2$ denote the steady-state variances of $v'$ and
$e'$ in Figure \ref{fig:lazo-con-linear-coder}, respectively.
\end{corollary}

\begin{proof}
Consider the LTI feedback system of Figure
\ref{fig:lazo-con-linear-coder} and recall the definition of
both $\phi_u'(D)$ and $\phi_w'$ in \eqref{eq:lower-bound} and
\eqref{eq:def-phi-w}. Since $D>D_{\inf}$, the problem of
finding $\phi_u'(D)$ is feasible (see Appendix~\ref{ap:factibilidad}).  Thus, for any $\epsilon>0$, there
exist a proper LTI filter $L_{\epsilon}$, and
$\sigma_{\epsilon'}^2\in\R^+$, such that $\sigma_{e'}^2\leq D$
and
\begin{align}
\phi_u'(D) + \epsilon & \geq \phi_w'(1,L_{\epsilon}, \sigma_{\epsilon'}^2),
\end{align}
where we have used the fact that $u'=w'$ whenever $F=1$.  On
the other hand, Lemma \ref{lema:w-blanco} guarantees that there
exists a pair of proper filters $(\bar{F}_{\epsilon},
\bar{L}_{\epsilon})$, with $\bar{F}_{\epsilon}$ biproper, such
that the auxiliary feedback system of Figure
\ref{fig:lazo-con-linear-coder} is internally stable and
well-posed,
\begin{align}
\left. \sigma_{e'}^2 \right|_{(F, L, \sigma_q^2)=(1, L_{\epsilon}, \sigma_{\epsilon'}^2 )}
    = \left. \sigma_{e'}^2 \right|_{(F, L, \sigma_q^2)
        =(\bar{F}_{\epsilon}, \bar{L}_{\epsilon}, \sigma_{\epsilon'}^2 ) } \leq D
\end{align}
and, in addition, such that for any $\eta>0$
\begin{align}\label{eq:clave-lower-bound-snr}
\phi_u'(D) + \epsilon + \eta =  \left. \frac{1}{2} \log{ \lp 1
                            + \frac{\sigma_{v'}^2}{\sigma_q^2} \rp } \right|_{(F, L, \sigma_q^2)
                                = ( \bar{F}_{\epsilon}, \bar{L}_{\epsilon}, \sigma_{\epsilon'}^2 ) }
                                     \geq \frac{1}{2}\log{\lp 1 + \gamma'(D) \rp},
\end{align}
where the inequality follows from the definition of
$\gamma'(D)$. Since \eqref{eq:clave-lower-bound-snr} holds for
any $\epsilon, \eta>0$, our claim is now immediate from Theorem
\ref{teo:lower-bound}. \fin
\end{proof}

Corollary \ref{coro:lower-bound-snr} shows that a lower bound
on $\mathscr{R}(D)$ can be obtained by first characterizing
$\gamma'(D)$, i.e., by first characterizing, for the auxiliary
LTI feedback system of Figure \ref{fig:lazo-con-linear-coder},
the minimal steady-state SNR $\gamma' =
\sigma_{v'}^2/\sigma_{q'}^2$ that guarantees that the
steady-state variance of the error signal $e'$ is upper bounded
by $D$.   Section \ref{sec:computing-the-bounds} discusses how
to obtain a numerical approximation to $\gamma'(D)$.

\section{An Upper Bound on $\mathscr{R}(D)$}\label{sec:upper-linear-coding}
This section shows that it is indeed possible to achieve any
distortion level $D\in(D_{\inf}, \infty)$ while incurring an
average data rate that exceeds the lower bound on
$\mathscr{R}(D)$ in Corollary \ref{coro:lower-bound-snr} by
less than (approximately) $1.254$ bits per sample.  

\begin{definition}\label{def:linear-source-coder}
The source coding scheme described by \eqref{eq:enc} and
\eqref{eq:dec} is said to be \emph{linear} if and only if, when
used around an error-free zero-delay digital channel, is such
that its input $y$ and output $u$ are related via
\begin{align}\label{eq:def-linear-coding}
u = F w, \quad w = q + v, \quad v = L \diag{ z^{-1}, 1} \matriz{w \\ y},
\end{align}
where $v$ and $w$ are scalar-valued auxiliary signals, $q$ is
an independent second-order zero-mean i.i.d.~sequence, and both
$F$ and $L$ are the transfer functions of proper LTI systems
that, together with the unit delay $z^{-1}$, have deterministic
initial states.\fin
\end{definition}

\begin{remark}
In Definition \ref{def:linear-source-coder}, the requirement of
$q$ being independent (without reference to other random
variables or processes) is to be understood as requiring $q$ to
be independent of all exogenous processes and initial states in
the (feedback) system in which the source coding scheme is
embedded. In particular, when an independent source coding
scheme is used in the NCS of Figure \ref{fig:ncs-general}, $q$
is to be assumed independent of $(x_o,d)$. \fin
\end{remark}

The class of linear source coding schemes is motivated by the
results of Section \ref{sec:lower} and  generalizes the class
of independent source coding schemes introduced in
\cite{sideos11-TAC}.\footnote{In the latter class, $u = y +
\Omega q$ with $\Omega \in\cal{U}_{\infty}$.}  We note that
independent source coding schemes do not necessarily satisfy
Assumption \ref{assu:coding-schemes}.

Linear source coding schemes are defined in terms of their
input-output relationship with no regard as to how the channel
input $y_c$ is related to the source coding scheme input $y$.
A simple way of making that relationship explicit is by using
an entropy-coded dithered quantizer (ECDQ; \cite{sideos11-TAC,
zamfed92}). When using such a device, $v$ and $w$ in
\eqref{eq:def-linear-coding}, and the channel input $y_c$ and
output $u_c$, are related via
\begin{subequations}\label{eq:ecdq}
\begin{align}
w(k) &= \hat{s}(k) - d_h(k), & \hat{s}(k) & = \mathscr{H}_k^{-1}(u_c(k), d_h(k)), \\
y_c(k) &= \mathscr{H}_k(s(k), d_h(k)), & \quad s(k) & = \cal{Q}(v(k)+d_h(k)),
\end{align}
\end{subequations}
where $d_h$ is a dither signal available at both the encoder
and decoder sides, $\cal{Q}:\R\to \{i\Delta; i\in\Z\}$ denotes
a uniform quantizer with step size $\Delta\in\R^+$,
$\mathscr{H}_k$ is a mapping describing an entropy-coder (i.e.,
a loss-less encoder \cite[Ch.5]{covtho06}) whose output symbol
is chosen according to the conditional distribution of $s(k)$,
given $d_h(k)$, and $\mathscr{H}^{-1}_k$ is a mapping
describing the entropy-decoder that is complementary to the
entropy-coder at the encoder side.

\begin{lemma}[Theorem 5.3 in \cite{sideos11-TAC}]\label{lema:ruido-en-ecdq}
Consider the setup of Figure~\ref{fig:ecdq-feedback}, where the
ECDQ is as in \eqref{eq:ecdq} and has a finite quantization
step $\Delta$.  Assume that $\bar{P}$ is a proper real rational
transfer function, that the open-loop transfer function from
$w$ to $v$ is single-input single-output and strictly proper,
and that the signal $\bar{d}$ is a white noise sequence jointly
second order with the initial state $\bar{x}_o$ of $\bar{P}$.
If the dither $d_h$ is i.i.d., independent of $(\bar{x}_o,
\bar{d})$ and uniformly distributed on $(-\Delta/2, \Delta/2)$,
then $w - v$ is i.i.d., independent of $(\bar{x}_o,\bar{d})$
and uniformly distributed in $(-\Delta/2, \Delta/2)$. \fin
\end{lemma}

\begin{figure}
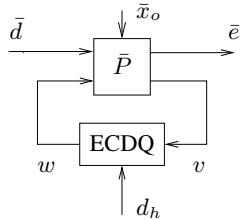

\centering
\input dq-feedback.pstex_t
\caption{Entropy coded dithered quantizer inside a feedback
loop.}\label{fig:ecdq-feedback}
\end{figure}

It follows that any coding scheme described by
\eqref{eq:def-linear-coding} and \eqref{eq:ecdq}, with dither
as in Lemma \ref{lema:ruido-en-ecdq}, is a linear source coding
scheme. Any such coding scheme will be referred to as an
ECDQ-based linear source coding scheme. Figure
\ref{fig:proposed-coding-scheme} depicts an ECDQ-based linear
source coding scheme where we have made explicit the fact that,
since the channel is error-free and has zero delay, $\hat{s}=s$
and, thus, $w$ can be obtained at the encoder side without
making use of any additional feedback channel.

The next lemma gives an upper bound on the (operational)
average data rate in an ECDQ-based linear source coding scheme.

\begin{lemma}\label{lema:tasa-en-ecdq}
Consider the NCS of Figure \ref{fig:ncs-general} and suppose
that that Assumption \ref{assu:G} holds.  Then, there exists an
ECDQ-based linear source coding scheme such that the resulting
NCS is AWSS.  For any such coding scheme,
\begin{align}
\mathscr{R} < \frac{1}{2} \log{\lp 1+ \frac{\sigma_v^2}{\sigma_q^2} \rp} +
    \frac{1}{2}\log{ \lp \frac{ 2\pi e }{12} \rp} + \log{2},
\end{align}
where $\sigma_v^2$ is the steady-state variance of the
auxiliary signal $v$, and $\sigma_q^2=\Delta^2/12$ is the
linear source coding scheme noise variance (see
\eqref{eq:def-linear-coding}).
\end{lemma}

\begin{proof}
Consider the NCS of Figure \ref{fig:ncs-general} and assume
that the source coding scheme is linear.   Since Assumption
\ref{assu:G} holds, there exist proper LTI filters $L$ and $F$
such that the resulting NCS is internally stable and well-posed
(one possibility is to choose $L$ such that $v=y$ and to pick
any $F$ which internally stabilizes $P$).  For any such choice
of filters, the open loop system linking $w$ with $v$ is
stabilizable with unity feedback.   Our claim now follows
immediately upon using Corollary 5.3 in \cite{sideos11-TAC} and
the description for the coding noise in Lemma
\ref{lema:ruido-en-ecdq}.\fin
\end{proof}

We are now in a position to prove the main results of this
section:

\begin{theorem}\label{teo:rate-distortion-linear}
Consider the NCS of Figure \ref{fig:ncs-general} and suppose
that Assumption \ref{assu:G} holds.   If $D\in ( D_{\inf},
\infty )$, then there exists an ECDQ-based linear source coding
scheme satisfying Assumption \ref{assu:coding-schemes} such
that the resulting NCS is AWSS, $\sigma_e^2\leq D$, and
\begin{align}\label{eq:upper-bound}
\mathscr{R} < \frac{1}{2} \log{ \Big( 1 +
\gamma'(D) \Big) } + \frac{1}{2}\log{ \lp \frac{
2\pi e }{12} \rp} + \log{2},
\end{align}
where $\gamma'(D)$ is as in \eqref{eq:lower-bound-snr}.
\end{theorem}

\begin{proof}
Since $D>D_{\inf}$, the problem in \eqref{eq:lower-bound-snr}
is feasible (see Appendix \ref{ap:factibilidad}).  Thus, there
exist proper LTI filters $L$ and $F$ rendering the feedback
system of Figure \ref{fig:lazo-con-linear-coder} internally
stable and well-posed, and $\sigma_{q'}^2\in\R^+$, such that,
in the scheme of Figure \ref{fig:lazo-con-linear-coder}, and
for any $\epsilon>0$, $\sigma_{e'}^2\leq D$ and
\begin{align}\label{eq:e-factible}
\frac{\sigma_{v'}^2}{\sigma_{q'}^2} \leq \gamma'(D) + \epsilon.
\end{align}
Denote the above choices for $L$, $F$ and $\sigma_{q'}^2$ by
$L_{\epsilon}, F_{\epsilon}$ and $\sigma_{\epsilon'}^2$,
respectively.  Given Lemma \ref{lema:w-blanco} and Jensen's
inequality, $F_{\epsilon}$ can be assumed to be biproper
without loss of generality.

\begin{figure}
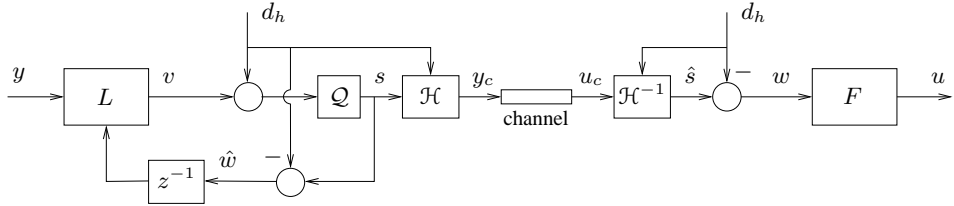
 \centering
\input ecdq_mas_filtros.pstex_t
\caption{Proposed source coding scheme.  If the channel is noiseless and delay free, then $\hat{s}=s$ and $\hat{w}=w$. }\label{fig:proposed-coding-scheme}
\end{figure}

Consider the NCS of Figure \ref{fig:ncs-general} and assume
that the link between $y$ and $u$ is given by an ECDQ-based
linear source coding scheme with parameters
$(L,F,\Delta)=(L_{\epsilon}, F_{\epsilon}, (
12\sigma_{\epsilon'}^2 )^{1/2})$, and set the initial states of
$L_{\epsilon}$, $F_{\epsilon}$ and of the channel feedback
delay to zero. The definition of $L_{\epsilon}$, $F_{\epsilon}$
and $\sigma_{\epsilon'}^2$, together with Lemma
\ref{lema:ruido-en-ecdq}, guarantee by construction that the
NCS that results from the above choice of coding scheme is AWSS
and that, in addition, the plant output $e$ and the auxiliary
signal $v$ in \eqref{eq:def-linear-coding} have steady-state
variances satisfying
\begin{align}\label{eq:clave-para-upper-bound}
\sigma_e^2 = \sigma_{e'}^2 \leq D, \quad \frac{\sigma_v^2}{\sigma_q^2} = \frac{\sigma_{v'}^2}{\sigma_{q'}^2} \leq \gamma'(D) + \epsilon.
\end{align}
By Lemma \ref{lema:tasa-en-ecdq} we also conclude that, for the
above described ECDQ-based linear source coding scheme, the
average expected length of the channel input $y_c$ satisfies,
for some suitable $\delta>0$,
\begin{multline}\label{eq:ineqs_R_ECDQ}
\mathscr{R}  < \frac{1}{2} \log \lp 1 +  \frac{\sigma_v^2}{\sigma_q^2} \rp + \frac{1}{2}\log{ \lp \frac{ 2\pi e }{12} \rp} + \log{2} - \delta \\
\leq \frac{1}{2} \log \Big( 1 +  \gamma'(D) + \epsilon \Big) +
\frac{1}{2}\log{ \lp \frac{ 2\pi e }{12} \rp} + \log{2} -
\delta
\end{multline}
where we have used \eqref{eq:clave-para-upper-bound}.  Thus,
inequality \eqref{eq:upper-bound} follows upon choosing a
sufficiently small $\epsilon>0$.

To complete the proof, we now show that the proposed source
coding scheme satisfies Assumption \ref{assu:coding-schemes}.
Except for the invertibility of the decoder, the properties of
$d_h$ guarantee that Assumption \ref{assu:coding-schemes} holds
(note that, in our case, $S_{\enc}=S_{\dec}=d_h$).  Since
$F_{\epsilon}$ is biproper and its initial state is
deterministic, knowledge of $u^k$ is equivalent to knowledge of
$w^k$.  If one now proceeds as in the proof of Corollary 5.1 in
\cite{sideos11-TAC}, it follows that one can recover $u_c^k$
from $w^k$ upon knowledge of $d_h^k$.  Thus, upon knowledge of
$S_{\dec}^k=d_h^k$, one can recover $u_c^k$ from $u^k$ and the
decoder is invertible as required.  \fin
\end{proof}

\begin{remark}
The proof of Theorem \ref{teo:rate-distortion-linear} is
constructive.  Indeed, it suggests a way to build a source
coding scheme that renders the resulting NCS of Figure
\ref{fig:ncs-general} AWSS, and achieves  $\sigma_e^2\leq D$
while incurring an average data rate that is upper bounded by
the right-hand side of \eqref{eq:upper-bound}. \fin
\end{remark}

Theorem \ref{teo:rate-distortion-linear} shows that the lower
bound on $\mathscr{R}(D)$ derived in Corollary
\ref{coro:lower-bound-snr} is tight up to $\frac{1}{2}\log{ \lp
\frac{ 2\pi e }{12} \rp} + \log{2}$ nats per sample (i.e.,
tight up to approximately $1.254$ bits per sample).   Whilst
the lower bound in Corollary \ref{coro:lower-bound-snr} was
derived by using an information-theoretic argument, the upper
bound in \eqref{eq:upper-bound} hinges on a specific source
coding scheme that uses suitably chosen LTI filters in
conjunction with an ECDQ.  It follows from the discussion in
Section V-B in \cite{sideos11-TAC} that the gap between the
derived upper and lower bounds on $\mathscr{R}(D)$ arises from
two facts:  First, ECDQs introduce a coding noise which is
uniform and not Gaussian (this amounts to the additional
$\frac{1}{2}\log{ \lp \frac{ 2\pi e }{12} \rp}$ nats per
sample).  Second, the proposed coding scheme works on a
sample-by-sample basis and practical entropy-coders are not
perfectly efficient \cite[Chapter 5]{covtho06} (this amounts to
an additional $\log{2}$ nats per sample).  
We emphasize, however, that the above gap corresponds to a
worst case gaps and it can be significantly smaller in practice
(see Section~\ref{sec:ejemplo}).

A key aspect of our results is that they are stated in terms of
the solution to the constrained SNR minimization problem in
\eqref{eq:lower-bound-snr}.  As such, they highlight the role
played by SNR constraints in networked control systems, and
thus complement, e.g., \cite{silpul11} where the connection
between SNR constraints and other communication constraints has
been explored.   As already mentioned before, a way of
obtaining a solution to the problem in
\eqref{eq:lower-bound-snr} will be discussed in Section
\ref{sec:computing-and-approximations} below.

%
%

\begin{remark}
It is well-known \cite{naieva04} that, when causal source
coding schemes of arbitrary complexity are employed, it is
possible to mean-square stabilize an LTI plant if and only if
the corresponding average data rate $\mathscr{R}$ is larger
than $\sum_{i=1}^{n_p} \log{\abs{p_i}}$, where $p_i$ denotes
the $i^{th}$ unstable plant pole.  On the other hand, it is
straightforward use Theorem 17 in \cite{sigoqu10-auto}, in
conjunction with the proof of Theorem
\ref{teo:rate-distortion-linear}, to show that any plant
satisfying Assumption \ref{assu:G} can be stabilized in the
sense of Definition \ref{def:estabilidad} by incurring an
average data rate that satisfies
\begin{align}
\mathscr{R} < \sum_{i=1}^{n_p} \log{\abs{p_i}}
+ \frac{1}{2}\log{ \lp \frac{ 2\pi e }{12} \rp} + \log{2}.
\end{align}
The above observation shows, for plants satisfying Assumption
\ref{assu:G}, that it suffices to use an ECDQ-based linear
source coding scheme to achieve stability at rates which are at
most $\frac{1}{2}\log{ \lp \frac{ 2\pi e }{12} \rp} + \log{2}$
nats per sample away from the absolute minimal average data
rate compatible with stability (see also
\cite{sideos11-TAC}).\fin
\end{remark}

\section{Computations and Approximate Implementation}\label{sec:computing-and-approximations}

\subsection{Computing the bounds on $\mathscr{R}(D)$}\label{sec:computing-the-bounds}
The bounds on $\mathscr{R}(D)$ presented in Corollary
\ref{coro:lower-bound-snr} and Theorem
\ref{teo:rate-distortion-linear} are functions of the minimal
SNR $\gamma'(D)$ in \eqref{eq:lower-bound-snr}.  In this
section, we show that the problem of finding $\gamma'(D)$ is
equivalent to an SNR constrained optimal control problem
previously addressed in \cite{sideos10-necsys, derost12,
johann11}.

To proceed, we first note that a straightforward manipulation
based on Figure \ref{fig:lazo-con-linear-coder} yields, for any
$\sigma_{q'}^2\in\R^+$ and any proper LTI filters $F$ and $L$
that render the LTI feedback system of Figure
\ref{fig:lazo-con-linear-coder} internally stable and
well-posed,
\begin{align}
\label{eq:gamma-explicito} \gamma' & = \norma{ S-1 }_2^2
    + \sigma_{q'}^{-2} \norma{L_yP_{21}S}_2^2 = \norma{S}_2^2 +\sigma_{q'}^{-2} \norma{L_y S P_{21} }_2^2 - 1,\\
\label{eq:performance-explicito} \sigma_{e'}^2 & = \norma{P_{11} + P_{12} K(1-P_{22}K)^{-1} P_{21} }_2^2
    + \norma{P_{12} FS }\sigma_{q'}^2,
\end{align}
where $S$ is as in \eqref{eq:def-S}, $K\treq FL_y \lp
1-L_wz^{-1} \rp^{-1}$, $L_w$ and $L_y$ are such that $L = [ \,
L_w \; L_y \,]$ (see \eqref{eq:linear-coding-explicito}), and
we have used that fact that, since $F$ and $L$ are internally
stabilizing, $L_w$ is proper and $P_{22}$ is assumed to be
strictly proper, $S$ is stable, $S(\infty)=1$ and hence
$\norma{S-1}_2^2 = \norma{S}_2^2-1$.

We now define, for the feedback system of Figure
\ref{fig:lazo-con-linear-coder}, the auxiliary problem of
finding
\begin{align}\label{eq:optimal-performance}
J'(\Gamma) \treq \inf_{\gamma'\leq \Gamma} \sigma_{e'}^2,
\end{align}
where the minimization is performed with respect to all proper
LTI filters $F$ and $L$, and auxiliary noise variances
$\sigma_{q'}^2\in\R^+$, that render the LTI feedback system of
Figure \ref{fig:lazo-con-linear-coder} internally stable and
well-posed.  Given our assumptions, if the plant $P$ is
unstable, then the problem in \eqref{eq:optimal-performance} is
feasible if and only if $\Gamma>\Gamma_{\inf}$, where
$\Gamma_{\inf}$ denotes the infimal SNR $\gamma'$ that is
compatible with mean-square stability in the feedback system of
Figure \ref{fig:lazo-con-linear-coder}  (see
\cite{sigoqu10-auto, frmibr07}).  If $P$ is stable, then the
problem in \eqref{eq:optimal-performance} is feasible if and
only if $\Gamma \geq \Gamma_{\inf}=0$.

\begin{lemma}\label{lema:activas}
Consider the problems of finding both $\gamma'(D)$  and
$J'(\Gamma)$ in \eqref{eq:lower-bound-snr} and
\eqref{eq:optimal-performance}, respectively.   Assume, in
addition, that the plant $P$ is such that $P_{21}\neq 0$ and
$P_{12}\neq 0$.
\begin{enumerate}
\item If $D\in (D_{\inf}, \infty)$, the plant is unstable, or
stable with $D < \norma{P_{11}}_2^2$, then $\gamma'(D)$ is
a strictly decreasing function of $D$ and the inequality
constraint in \eqref{eq:lower-bound-snr} is active at the
optimum.
\item If $\Gamma>\Gamma_{\inf}$, then $J'(\Gamma)$ is a strictly decreasing
function of $\Gamma$ and the inequality constraint in
\eqref{eq:optimal-performance} can be assumed to be active
at the optimum without loss of generality.
\end{enumerate}
\end{lemma}

\begin{proof}
\begin{enumerate}
\item Consider the definition of $\gamma'(D)$ in \eqref{eq:lower-bound-snr}
and define $\kappa\treq \norma{P_{11} + P_{12}
K(1-P_{22}K)^{-1} P_{21} }_2^2$.  We first show that our
assumptions imply that $D-\kappa > 0$ at the optimum.
(Since $D>D_{\inf}$, $D-\kappa$ is always non negative for
any feasible set of parameters $F,L$ and
$\sigma_{q'}^2>0$.) Indeed, assume on the contrary that
$D-\kappa=0$ at the optimum. Given
\eqref{eq:performance-explicito}, this would imply that
$\sigma_{q'}^2=0$ or $P_{12}FS=0$ at the optimum. Given our
assumptions and the definition of $\gamma'(D)$, only $F=0$
is possible.  However, $F=0$ is not compatible with
internal stability when the plant is unstable.  In the
stable plant case, $F=0$ implies $\sigma_{e'}^2 =
\norma{P_{11}}_2^2$ (note that $F=0\implies K=0$).  The
latter equality is however unfeasible since, by assumption,
$D < \norma{P_{11}}_2^2$.

Given the above, if $D>D_{\inf}$, and the plant is unstable
or is stable with $D < \norma{P_{11}}_2^2$, then $D-\kappa
> 0 $ at the optimum. Hence, the performance constraint
$\sigma_{e'}^2\leq D$ in the definition of $\gamma'(D)$ is
equivalent to
\begin{align}\label{eq:cota-inferior-1/sigmaq2}
\frac{1}{\sigma_{q'}^2} \geq \frac{\norma{P_{12}FS}_2^2}{D - \kappa}
\end{align}
at the optimum (see \eqref{eq:lower-bound-snr} and
\eqref{eq:performance-explicito}). Since $\gamma'$ is a
nondecreasing function of $\sigma_{q'}^{-2}$, it follows
that the optimal choice for $\sigma_{q'}^2$ is such that
the inequality constraint is active at the optimum.

We now show that $\gamma'(D)$ is strictly decreasing in
$D$. Our assumptions guarantee that $K\neq 0$ and hence
$FL_y\neq 0$.  By using \eqref{eq:cota-inferior-1/sigmaq2}
in \eqref{eq:gamma-explicito}, and the fact that the
optimal choice for $\sigma_{q'}^2$ achieves equality in
\eqref{eq:cota-inferior-1/sigmaq2}, the result follows
immediately.
\item Consider the definition of $J'(\Gamma)$ in
\eqref{eq:optimal-performance}.   By using an argument
similar the one used in Part 1 above, it follows that our
assumptions imply that one can assume, without loss of
generality, that $\Gamma+1-\norma{S} > 0$ at the optimum
(see also \cite[pages 103--104]{johann11}).
Our claims now follow by proceeding as in Part 1)
above.\fin
\end{enumerate}
\end{proof}

Lemma \ref{lema:activas} shows, for almost all cases of
interest,\footnote{If $\Gamma=\Gamma_{\inf}$, then either the
problem of finding $J'(\Gamma)$ is unfeasible (unstable plant
case) or $\Gamma=0$ (stable plant case).  In the latter case,
no information can be conveyed through the channel. On the
other hand, if the plant is stable and $D\geq
\norma{P_{11}}_2^2$ then $\gamma'(D)=0$ and it is optimal to
leave the plant in open loop.  The above cases are clearly
uninteresting and have thus been omitted from the discussion in
Lemma \ref{lema:activas}.} that the inequality constraints in
the optimization problems defining both $\gamma'(D)$ and
$J'(\Gamma)$ can be assumed to be active at the optimums,
without loss of generality.  This fact is exploited below to
relate the solutions to these problems.

\begin{theorem}\label{teo:problemas-equivalentes}
Consider the optimization problems defining both $\gamma'(D)$
and $J'(\Gamma)$ in \eqref{eq:lower-bound-snr} and
\eqref{eq:optimal-performance}, respectively.   Assume that
$\Gamma>\Gamma_{\inf}$, $D>D_{\inf}$, that the plant $P$ is
such that $P_{12}\neq 0$ and $P_{21}\neq 0$, and that, if $P$
is stable, then $D<\norma{P_{11}}_2^2$ holds. Then,
\begin{align}
D = J'(\gamma'(D)), \quad \Gamma = \gamma'(J'(\Gamma)).
\end{align}
\end{theorem}

\begin{proof}
We will only prove that $\Gamma=\gamma'(J'(\Gamma))$.  Our
remaining claim follows by using a similar argument. Since
$\Gamma>\Gamma_{\inf}$, the problem of finding $J'(\Gamma)$ is
feasible. Thus, for any $\epsilon>0$, there exist proper LTI
filters $F_{\epsilon}$ and $L_{\epsilon}$, and
$\sigma_{\epsilon'}^2\in\R^+$, that render the system of Figure
\ref{fig:lazo-con-linear-coder} internally stable and
well-posed, and guarantee that
\begin{align}\label{eq:para-demo-equivalencias}
\evalnoll{\sigma_{e'}^2}{(F,L, \sigma_{q'}^2)}{(F_{\epsilon}, L_{\epsilon}, \sigma_{\epsilon'}^2)}
    \leq J'(\Gamma) + \epsilon,
        \quad \evalnoll{\gamma'}{(F,L,\sigma_{q'}^2)}
            {( F_{\epsilon}, L_{\epsilon}, \sigma_{\epsilon'}^2)} = \Gamma,
\end{align}
where we have used Lemma \ref{lema:activas} to write an
equality in the SNR constraint.  Since the inequality in
\eqref{eq:para-demo-equivalencias} is valid for any
$\epsilon>0$, it follows that there exist a feasible point for
the problem of finding $\gamma'(J'(\Gamma))$ and, in addition,
that $\gamma'(J'(\Gamma))\leq \Gamma$.  The proof of our second
claim would follow if we show that $\gamma'(J'(\Gamma)) <
\Gamma$ is impossible.    Assume that $\gamma'(J'(\Gamma)) <
\Gamma$ is indeed true.  Then, there exist decision variables
such that $\gamma'=\breve{\gamma} < \Gamma$ and $\sigma_{e'}^2
= J'(\Gamma)$.  (Again, we use Lemma \ref{lema:activas} to
write an equality in the constraint defining
$\gamma'(J'(\Gamma))$.)  Thus, we conclude that
\begin{align}
\inf_{\gamma' = \breve{\gamma} } \sigma_{e'}^2  = J'(\breve{\gamma}) \leq J'(\Gamma),
\end{align}
where the first equality follows from the fact that the
constraint is active at the optimum when calculating
$J'(\breve{\gamma})$.  The above inequality contradicts the
fact that, given our assumptions, Lemma \ref{lema:activas}
guarantees that $J'(\Gamma)$ is a strictly decreasing function
of $\Gamma$.  The proof is thus completed.\fin
\end{proof}

Theorem \ref{teo:problemas-equivalentes} shows, for almost all
cases of interest, that the problem of finding $\gamma'(D)$ in
\eqref{eq:lower-bound-snr} is equivalent to that of finding
$J'(\Gamma)$ in \eqref{eq:optimal-performance}. The latter
problem was shown to be equivalent to a convex problem in
\cite{sideos10-necsys}.  For doing so, \cite{sideos10-necsys}
showed that the problem of finding $J'(\Gamma)$ is equivalent
to the open-loop causal rate-distortion problem which was shown to be convex in~\cite{derost12}.  Shortly thereafter, the convexity of the
SNR constrained optimal control problem in
\eqref{eq:optimal-performance} was re-derived independently in
\cite{johann11}, where a formulation more amenable for numerical
computations is presented.  
We will thus not delve into the details on how to numerically
find $\gamma'(D)$ here, and refer the interested reader to
Section 3.3. in \cite{johann11} for details.

\subsection{Approximating the behavior of an ECDQ in practice}\label{sec:approx}
The previous subsection explained how a numerical
characterization of $\gamma'(D)$ can be obtained. Here, we will
briefly comment on the implementation of a source-coding scheme
which achieves the desired level of performance $D$, while
incurring an average data rate $\mathscr{R}$ satisfying
\eqref{eq:upper-bound}.  In principle, such coding scheme can
be designed as follows (see proof of Theorem
\ref{teo:rate-distortion-linear}):
\begin{itemize}
\item Use the procedure in \cite{johann11} and Theorem \ref{teo:problemas-equivalentes}
to find the filters $L$ and $F$, and the auxiliary noise
variance $\sigma_{q'}^2$, which solve the problem of
finding $\gamma'(D)$.
\item Use these filters in the ECDQ-based linear source coding
scheme of Figure \ref{fig:proposed-coding-scheme} and set
all initial states to zero.  Choose the ECDQ quantization
step as $\Delta = \lp 12\sigma_{q'}^2 \rp^{1/2}$, an i.i.d.
dither signal uniformly distributed on $(-\Delta/2,
\Delta/2)$ and independent of $(x_o, d)$, and appropriate
entropy coder and decoder mappings $\mathscr{H}_k$ and
$\mathscr{H}_k^{-1}$ (using, for instance, the Huffman
algorithm \cite{covtho06}).
\end{itemize}

Implementing an ECDQ requires the availability of the dither at
both the encoder and the decoder sides. Additionally, the
entropy coder $\mathscr{H}_k$  needs to generate a binary word
for each input value according to the conditional probability
of that input, given the current dither value.  The above
requirements are impossible to meet exactly in practice.
Indeed, the first one is tantamount to requiring an additional
perfect channel for being able to communicate the dither from
the encoder to the decoder. The second one would require an
uncountable number of dictionaries \cite{covtho06}, one for
each dither value.

Leaving finite range and precision issues aside, the behavior
of an ECDQ can be approximated in practice by using
synchronized uniformly distributed pseudo-random dither
sequences, generated at both the encoder and the decoder from
the same seed, and using entropy-coders and decoders which work
conditioned upon a uniformly quantized version of the dither.
By using such an approach, all signals in the NCS of Figure
\ref{fig:ncs-general}, except for the channel input and output,
will have the same statistics as if an ideal ECDQ was employed.
For each possible quantized dither value, one can build the
corresponding conditional dictionary in $\mathscr{H}_k$ by
using, for example, the Huffman coding
algorithm~\cite{covtho06}. The conditional statistics of the
quantizer outputs needed for this purpose, can be approximated
by the corresponding stationary statistics which can be estimated
empirically by simulation. 

\section{A numerical example}\label{sec:ejemplo}
%
%
%
%

\begin{figure}[htpb]
\centering
{\includegraphics[scale=0.55]{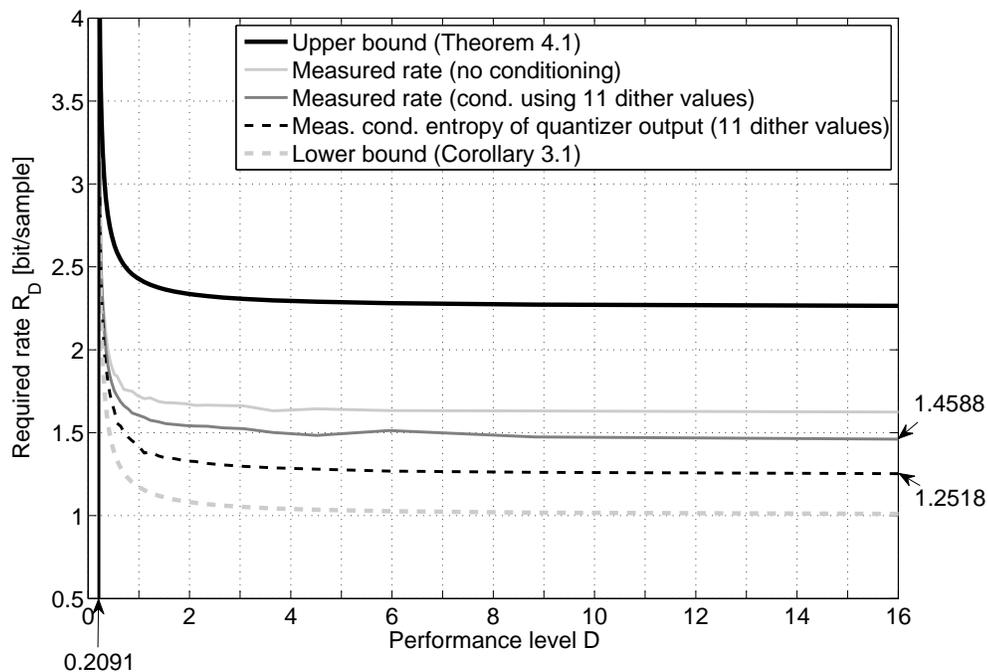}}
\caption{Bounds on the minimal average data rate required to attain a given closed-loop performance level.}
\label{fig:example_1}
\end{figure}

Assume that the plant $P$ in Figure \ref{fig:ncs-general} is
such that
\begin{equation}
y = \dfrac{0.165}{(z-2)(z-0.5789)}(u + d), \quad e = y
\end{equation}
and that $(x_o, d)$ is Gaussian with $d$ being unit-variance
white noise.  
By using the results of Sections \ref{sec:lower}
and \ref{sec:upper-linear-coding} we computed upper and lower
bounds on $\mathscr{R}(D)$ for several values of $D > D_{\inf}
= 0.2091$.   We also simulated an actual ECDQ-based linear
source coding scheme for each considered value for $D$. 
To that end, we followed the suggestions at the end of Section
\ref{sec:approx} and simulated ECDQs where the dither is
uniform and perfectly known at both ends of the channels, 
and where the entropy-coders work conditioned upon a quantized
version of the dither.  The results are presented in Figure
\ref{fig:example_1}.  In that figure we plot our upper and
lower bounds, and several other curves which report simulation
results.  
All simulation results (referred to as ``Measured''
in Figure \ref{fig:example_1}) are averages over twenty $10^4$-samples-long realizations.   
In particular, ``Measured rate (no
conditioning)'' corresponds to the average data rate in a case
where an empirically-tuned entropy-coder is employed which does
not make use of the knowledge of the dither values.  Even in
this case our upper bound proves to be rather loose. The curve
``Measured rate (cond. using 11 dither values)'' corresponds to
the rate achieved when using and entropy coder that works
conditioned upon $11$ uniformly-quantized dither values.   As
expected our results show that conditioning reduces the
incurred average data-rate.  (Simulations suggested that using
more than $11$ quantized dither values brings only negligible benefits in
terms of rate reduction.)   The curve ``Measured entropy of
quantizer output (11 dither values)'' corresponds to an
empirical estimate of the conditional entropy of the quantizer
output $s$, given the quantized dither values.

Our results show that our upper bound is loose.  This is
consistent with the fact that our upper bound was derived by
using worst case considerations.  The gap between the measured
rate (with conditioning) and our lower bound is about $0.45$
bits per sample, which is smaller than the worst case gap
$\frac{1}{2}\log\lp \frac{2\pi e}{12} \rp + \log 2$ nats per
sample (about $1.254$ bits per sample). On the other hand the
estimated conditional entropy of the quantizer output, given
the dither values, is about $0.25$ bits per sample above the
lower bound.  This implies that, in our simulations, the $0.46$
bit per sample gap is composed by about $0.21$ bits per sample
due to the inefficiency of the considered entropy coders, and
by about $0.25$ bits per sample due to the fact that the ECDQ
generates uniform and not Gaussian noise.

Our results show that, as expected, achieving a closed loop
performance arbitrarily close to the best non networked
performance $D_{\inf}$ requires arbitrarily high data rates.
Interestingly, however, for this example, it suffices to use
less that $3$ bits per sample to achieve a performance that is
essentially identical to the best non networked performance. It
is also interesting to observe that our bounds, and the
measured average data-rates, converge rapidly as $D\to \infty$.
Thus, whilst achieving an average data-rate arbitrarily close
to the minimal rate for stabilization severely compromises
performance \cite{naieva04}, our results suggest that the
performance loss incurred when forcing the average data rate to
be low might be modest in some cases.

\section{Conclusions}\label{sec:conclusions}
This paper has studied networked control systems subject to
average data rate constraints.  In particular, we have obtained
a characterization of the minimal average data rate that
guarantees a prescribed level of performance.  Our results have
been derived for LTI plants that have one scalar control input,
one scalar sensor output, and that are subject to Gaussian
disturbances and initial states.  No constraints besides
causality have been imposed on the considered source coding
schemes which yielded a universal lower bound on the minimal
average data rate that guarantees a given performance level.
Such bound was derived by noting that optimal performance-rate
tradeoffs can be described by source coding schemes that behave
like a set of LTI filters and a source of additive white noise.
Such insight was then used as motivation for building a source
coding scheme capable of achieving rates which are less than
$\frac{1}{2}\log\lp \frac{2\pi e}{12} \rp + \log 2$ nats per
sample away from our derived lower bound, while satisfying the
desired performance level constraint. Such coding schemes are
based upon entropy dithered quantizers and constitute
conceptually simple coding schemes.  A numerical example has
been include to illustrate our proposal.

Future work should focus on multiple input and multiple output
plant models, multichannel architectures, and on ways of
reducing the gap between the derived upper and lower bounds on
the minimal average data-rate that guarantees a given
performance level.



\appendix

\setcounter{lemma}{0}
\renewcommand{\thelemma}{\Alph{subsection}.\arabic{lemma}}

\setcounter{theorem}{0}
\renewcommand{\thetheorem}{\Alph{subsection}.\arabic{lemma}}

\section{Appendix}

\subsection{Three consequences of assuming $D>D_{\inf}$}\label{ap:factibilidad}
In this appendix we show that $D>D_{\inf}$ is sufficient for
the optimization problems in \eqref{eq:def-RD},
\eqref{eq:lower-bound} and \eqref{eq:lower-bound-snr} to be
feasible.  We will make extensive use of the definition of
$\gamma'(D)$ in \eqref{eq:lower-bound-snr}, the related
equations in \eqref{eq:gamma-explicito} and
\eqref{eq:performance-explicito}, the conventions regarding the
feedback scheme of Figure \ref{fig:lazo-con-linear-coder} made
on the paragraph preceding Theorem \ref{teo:lower-bound}, and
the definitions of ECDQs and ECDQ-based linear source coding
schemes made in Section \ref{sec:upper-linear-coding}.
\begin{figure}
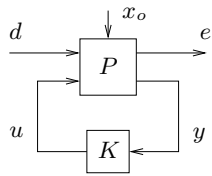

\centering
\input lazo_estandar.pstex_t
\vspace{-4mm}
\caption{Standard one-degree-of-freedom feedback loop around the plant $P$.}\label{fig:lazo-estandar}
\vspace{-8mm}
\end{figure}

Consider the feedback system of Figure \ref{fig:lazo-estandar},
where $P$, $d$ and $x_o$ satisfy Assumption \ref{assu:G}, and
the controller $K$ is such that $u(k)=\mathscr{K}_k (y^k)$ for
some arbitrary mappings $\mathscr{K}_k$.  Since $P$ is LTI and
all the involved random variables are Gaussian, it follows from
well-known results \cite{astrom70} that
\begin{align}
D_{\inf} = \inf_{K\in\cal{S}} \sigma_e^2,
\end{align}
where $\cal{S}$ is the set of all proper LTI filters which
render the the feedback system of Figure
\ref{fig:lazo-estandar} internally stable and well-posed.  Our
assumptions on $P$ guarantee that the above problem is
feasible.

The fact that $D>D_{\inf}$
and that the problem of finding $D_{\inf}$ is feasible, implies that for every $\epsilon\in(0,D-D_{\inf})$, there exists $K_{0}\in\cal{S}$
such that, in Figure~\ref{fig:lazo-estandar},
$\sigma^{2}_{e_{0}}\eq \evalnoll{\sigma_e^2}{K}{K_{0}}\leq D_{\inf}+\epsilon <D$.   
Consider the feedback
scheme of Figure \ref{fig:lazo-con-linear-coder} with $F=1$ and
$L=L_{0}$, where $L_{0}$ is such that $v = K_{0} \,
y$. 
Since $K_{0}\in\cal{S}$, the above choice renders the
feedback system of Fig.~\ref{fig:lazo-con-linear-coder}
internally stable and well-posed for any additive noise
variance $\sigma_{q'}^2\in\R^+$.  
This means that the resulting variance of $v'$, say $\sigsq_{v0}$, will be finite.
It also means that if $q'$ in Fig.~\ref{fig:lazo-con-linear-coder} 
is zero-mean AWGN with variance $\sigma_{q'}$, then the variances of $e'$ and of $v'$ will increase to 
$\sigsq_{e_{0}} + \beta_{0}^{(e)} \sigma_{q'}$, 
and
$\sigsq_{v_{0}} +\beta^{(v)}_{0}\sigsq_{q'}$, respectively,
for some finite factors $\beta^{(e)}_{0},\beta^{(v)}_{0}\geq 0$ which depend only upon $K_{0}$.
As a consequence, 
for every $D>D_{\inf}$, there exists $K_{0}\in\cal{S}$ such that 
in Fig.~\ref{fig:lazo-con-linear-coder} and for the above choice of
filters, 
$\evalnoll{\sigma_{e}'^2}{(F,L,\sigma_{q'}^2)}{(1,L_{0}, \sigsq_{q} )} < D_{\inf} + \tfrac{2}{3}(D-D_{\inf})$,  
by picking $\epsilon=(D-D_{\inf})/3$ and $q'$ as zero mean AWGN with variance $\sigsq_{q}= \tfrac{1}{3}(D-D_{\inf})/\beta_{0}$.
It immediately follows that the
above choice of parameters is also such that, in Figure
\ref{fig:lazo-con-linear-coder},
\begin{align}\label{eq:gammaD-factible}
\evalnoll{\gamma'}{(F,L,\sigma_{q'}^2)}{(1, L_{0},
\sigma_{q}^2)} 
<
\frac
{
3\beta_{0}^{(e)}\beta_{0}^{(v)}\sigsq_{v_{0}}
}
{D-D_{\inf} }
< \infty.
\end{align}
The latter inequality shows that the problem of finding
$\gamma'(D)$ in~\eqref{eq:lower-bound-snr} is feasible for any
$D>D_{\inf}$ (indeed, feasible while yielding all signals in the system jointly Gaussian). 
Now, from Jensen’s inequality and the concavity
of $\log$, it also follows from \eqref{eq:gammaD-factible} that
the problem of finding $\phi_u'(D)$ in \eqref{eq:lower-bound}
is feasible for any $D>D_{\inf}$.

We end this section by showing that the problem of finding
$\mathscr{R}(D)$ in \eqref{eq:def-RD} is also feasible when
$D>D_{\inf}$.  To that end, it suffices to consider an
ECDQ-based linear source coding scheme to link $y$ and $u$ in
Figure \ref{fig:ncs-general}, with parameters $\Delta= \lp
12\sigma_{\inf'}^2 \rp^{1/2}$, $F=1$ and $L=L_{\inf}$. Indeed,
by exploiting the properties of the latter choice of parameters
(see preceding paragraph), our claim follows by proceeding as
in the proof of Theorem \ref{teo:rate-distortion-linear} to
show that the above defined ECDQ-based linear source coding
scheme satisfies Assumption \ref{assu:coding-schemes}, renders
the NCS of Figure~\ref{fig:ncs-general} AWSS, and achieves
$\sigma_e^2<D$ at a finite average data rate $\mathscr{R}$.

\subsection{Auxiliary information-theoretic definitions and results}\label{ap:info-theory}
The following definitions and facts are standard and, unless
otherwise stated, can be found in \cite{covtho06}.   We assume
all random variables to have well defined (joint) probability
density functions (pdfs). The pdf of $x$ ($x,y$) is denoted
$f(x)$ ($f(x,y)$). $f(x|y)$ refers to the conditional pdf of
$x$, given $y$. $\meancr{x}{\cdot}$ denotes mean with respect
to the distribution of $x$.

The differential entropy of $x$ is defined via $h(x) \triangleq
-\meancr{x}{\log{f(x)}}$.  The conditional differential entropy
of $x$, given $y$, is defined via $h(x|y) \triangleq
-\meancr{x,y}{\log{f(x|y)}}$.  The mutual information between
two random variables $x$ and $y$ is defined via $I(x;y)\treq
-\meancr{x,y}{\log{ \lp f(x)f(y) / f(x,y) \rp }}$.  The
conditional mutual information between $x$ and $y$, given $z$,
is defined via $I(x;y|z) \treq I(x,z; y) - I(z;y)$.  The
following are properties of the above quantities:
\begin{enumerate}[({Property} 1)]
\item $I(x;y |z) = h(x|z) - h(x|y,z)$.
\item   If $f$ is a deterministic function, then
$h(x+f(y)|y)=h(x|y)$.
\item  If $x$ and $y$ are independent, then
$h(x|y)=h(x)$.
\item  $h(x_0, \cdots, x_{n-1}) = \sum_{i=0}^{n-1} h(x_i|
x_0, \cdots, x_{i-1})$, where $x_{-1}$ can be taken to be a deterministic constant, in which case $h(x_0|x_{-1})=h(x_{0})$.
\end{enumerate}

\begin{lemma}\label{lema:gauss-minimiza}
Assume that $(x,z)$ are jointly second-order random variables,
$x$ is Gaussian, and $z$ is arbitrarily distributed. Then,
$I(x;z)\geq (x;z_G)$, where $z_G$ is such that $(x,z_G)$ are
jointly Gaussian and have the same first- and second-order
(cross-) moments as $(x,z)$.
\end{lemma}

\begin{proof}
Immediate from the proof of Lemma 1 in \cite{derost08}.\fin
\end{proof}

\begin{figure}
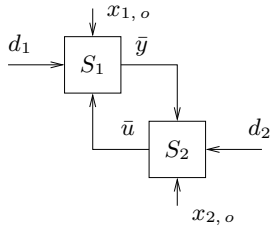

\centering
\input lazo_general.pstex_t
\vspace{-4mm}
\caption{Generic feedback system.}\label{fig:feedback-generico}
\vspace{-8mm}
\end{figure}

\begin{lemma}\label{lema:directed-estado}
Consider the generic feedback system of Figure
\ref{fig:feedback-generico}, where $S_1$ is an arbitrary (hence
possibly nonlinear and time-varying) causal dynamic system with
initial state $x_{1, \, o}$ and disturbance $d_1$, such that
$\bar{y}(k)=S_{1,k}(x_{1, \, o}, d_1^k, \bar{u}^{k-1})$ for
some (possibly nonlinear and time-varying) deterministic
mapping $S_{1,k}$, and $S_2$ is an arbitrary causal dynamic
system with disturbance $d_2$ and initial state $x_{2,o}$, such
that $\bar{u}(k)=S_{2,k}(x_{2,o}, d_2^k, \bar{y}^k)$ for some
(possibly nonlinear and time-varying) deterministic mapping
$S_{2,k}$.  If $(x_{2,o}, d_2)$ are jointly independent of $(
x_{1, \, o}, d_1)$, then
\begin{align}\label{eq:information-flux}
\sum_{i=0}^{k-1} I( \bar{u}(i); \bar{y}^i | \bar{u}^{i-1} ) = I( x_{1, \, o}, d_1^{k-1}; \bar{u}^{k-1} ).
\end{align}
\end{lemma}

\begin{proof}
Immediate from~\cite[Theorem~1]{desios13}. \fin
\end{proof}

Lemma \ref{lema:directed-estado} corresponds to a stronger
version of Theorem 5.1 in \cite{mardah05}.  Indeed, the latter
result makes use of additional assumptions on system $S_2$,
does not take side information into account, and only shows
that the left hand side of \eqref{eq:information-flux} is lower
bounded by the corresponding right-hand side.

 \bibliographystyle{\BibPath/IEEEtran}
 \bibliography{\BibPath/IEEEabrv,%
\BibPath/bibliografia2%
}

\end{document}